\documentclass[reqno,11pt]{amsproc}
\usepackage{amssymb}
\usepackage{amsmath}
\usepackage{amsthm}
\usepackage{amsfonts}
\usepackage{microtype}
\usepackage{xcolor}
\usepackage[left=2.5cm,right=2.5cm]{geometry}
\usepackage{mathrsfs}
\usepackage{mathtools}
\usepackage{thmtools}
\usepackage{cancel}

\usepackage[backend=biber,style=ieee,sorting=none,backref=true,citestyle=numeric-comp,giveninits=false]{biblatex}
\usepackage{tikz,tikz-cd}
\usetikzlibrary{arrows,cd,decorations.markings,arrows.meta}

\definecolor{myurlcolor}{rgb}{0,0,0.4}
\definecolor{mycitecolor}{rgb}{0,0.5,0}
\definecolor{myrefcolor}{rgb}{0.5,0,0}
\usepackage[breaklinks=true]{hyperref}
\hypersetup{colorlinks,
linkcolor=myrefcolor,
citecolor=mycitecolor,
urlcolor=myurlcolor,
hypertexnames=false}
\usepackage[capitalize]{cleveref}

\theoremstyle{plain}
\declaretheorem[name=Dummy,numberwithin=section]{dummy}
\declaretheorem[name=Theorem,sibling=dummy]{thm}
\declaretheorem[name=Lemma,Refname={Lemma,Lemmas},sibling=dummy]{lem}
\declaretheorem[name=Proposition,sibling=dummy]{prop}
\declaretheorem[name=Corollary,sibling=dummy]{cor}
\declaretheorem[name=Definition,sibling=dummy]{defn}
\declaretheorem[name=Question,sibling=dummy]{qstn}

\theoremstyle{remark}
\declaretheorem[name=Example,sibling=dummy]{ex}
\declaretheorem[name=Remark,sibling=dummy]{rem}

\numberwithin{equation}{section}
\Crefname{equation}{}{}		

\allowdisplaybreaks

\let\originalleft\left
\let\originalright\right
\renewcommand{\left}{\mathopen{}\mathclose\bgroup\originalleft}
\renewcommand{\right}{\aftergroup\egroup\originalright}

\usepackage{enumitem}
\setlist[enumerate]{label=(\roman*),itemsep=5pt,topsep=8pt}
\setlist[itemize]{label=$\triangleright$,itemsep=5pt,topsep=6pt}
\Crefformat{enumi}{#2#1#3}

\newcommand{\newterm}[1]{\textbf{#1}}

\newcommand{\beq}{\begin{equation}}
\newcommand{\eeq}{\end{equation}}

\newcommand{\R}{\mathbb{R}}
\newcommand{\C}{\mathbb{C}}

\newcommand{\id}{\mathsf{id}}		

\newcommand{\FinSet}{\mathsf{FinSet}}
\newcommand{\FinStoch}{\mathsf{FinStoch}}
\newcommand{\Set}{\mathsf{Set}}
\newcommand{\Conv}{\mathsf{Conv}}

\newcommand{\BoolAlg}{\mathsf{BoolAlg}}
\newcommand{\LT}{\mathsf{LT}}
\newcommand{\measfun}{\mathcal{M}}	
\newcommand{\measfunb}{\mathcal{N}}	

\newcommand{\tr}{\mathrm{tr}}

\usepackage{hyphenat}
\usepackage{todonotes}

\begin{document}



\title{Why measurements are made of effects}

\author{Tobias Fritz}

\address{Department of Mathematics, University of Innsbruck, Austria}
\email{tobias.fritz@uibk.ac.at}


\thanks{\textit{Acknowledgements.} We thank Tom\'a\v{s} Gonda, Ryshard-Pavel Kostecki and Sam Staton for discussion and a number of interesting suggestions.
	The author would like to thank the Isaac Newton Institute for Mathematical Sciences
for the support and hospitality during the programme \emph{Causal inference: From theory to
practice and back again} when work on this paper was undertaken. This work was
supported by: EPSRC Grant Number EP/Z000580/1.}

\begin{abstract}
	Both in quantum theory and in general probabilistic theories, measurements with $n$ outcomes are modelled as $n$-tuples of \emph{effects} summing up to the unit effect. Why is this the case, and can this assumption be meaningfully relaxed? Here we develop \emph{generalized measurement theories (GMTs)} as a mathematical framework for physical theories that is complementary to general probabilistic theories, and where this kind of question can be made precise and answered. We then give a definition of \emph{probabilistic state} on a GMT, prove that measurements are made of effects in every GMT in which the probabilistic states separate the measurements, and also argue that this separation condition is physically well-motivated. Finally, we also discuss when a GMT should be considered classical and characterize GMTs corresponding to Boolean algebras as those that are strongly classical and projective.
\end{abstract}

\maketitle

\tableofcontents

\section{Introduction}

Every physical theory has a notion of measurement.
Such a notion describes all the theoretically possible ways of gathering information about the physical system under consideration.
In quantum theory, there are \emph{two} main notions of measurement: PVMs (projective measurements) and POVMs (positive operator-valued measures).
Then a POVM with outcomes in a finite set $X$ is a family of positive operators
\begin{equation}
	\label{povm}
	(P_x)_{x \in X}
\end{equation}
summing up to the identity operator,
\begin{equation}
	\label{povm_normalized}
	\sum_{x \in X} P_x = 1.
\end{equation}
Similar formalizations of measurements are commonly used in the more permissive framework of \newterm{general probabilistic theories (GPTs)}, also known as \emph{convex operational theories}~\cite{Ludwig_axiomatic,Barrett,BW,Plavala}.
To spell this out, recall that a GPT can be defined as a triple $(V,\mathcal{S},\tr)$ consisting of a finite-dimensional\footnote{The extension to the infinite-dimensional case is feasible using standard structures from functional analysis, as already considered in 1970 by Davies and Lewis~\cite{DL} and independently by Ludwig~\cite{Ludwig_deutung,Ludwig_axiomatic}, with closely related structures considered even earlier by Gunson~\cite{Gunson}. We will consider a variant of Ludwig's approach in \cref{gpt}.} real vector space $V$, a compact convex set of states $\mathcal{S} \subseteq V$ and a normalization map
\[
	\tr \: \colon \: V \longrightarrow \R
\]
which is linear and satisfies $\tr(s) = 1$ for all $s \in \mathcal{S}$.
This definition makes clear that the states are the primary structure of a GPT and measurements are a derived concept.
Under the \newterm{no-restriction hypothesis}~\cite{CDP}, one defines an \newterm{effect} to be a linear functional $e : V \to \R$ which is nonnegative on $\mathcal{S}$ and bounded above by $\tr$, meaning that
\[
	0 \leq e(s) \leq 1 \qquad \forall s \in \mathcal{S}.
\]
A measurement in the GPT with outcomes in a finite set $X$ is then a family of effects
\begin{equation}
	\label{effect_family}
	(e_x)_{x \in X}
\end{equation}
such that
\begin{equation}
	\label{effect_normalization}
	\sum_{x \in X} e_x = \tr,
\end{equation}
where the sum is in the dual space $V^*$.
By imposing various kinds of additional conditions on the effects or on how measurements can be combined to effects, one can alternative notions of measurement which violate the no-restriction hypothesis while still being a priori reasonable on physical grounds~\cite{FGHL}.

\begin{ex}
	In quantum theory with $n$-dimensional Hilbert space, $C$ is the cone of positive semidefinite matrices in the matrix algebra $M_n(\C)$, the normalization functional $\tr$ is the trace, $\mathcal{S}$ corresponds to the set of density matrices, and the no-restriction hypothesis implies that the measurements are precisely the POVMs~\eqref{povm}.

	On the other hand, introductory quantum mechanics typically only introduces \emph{projective} measurements.
	It seems reasonable to consider these as forming a separate system of measurements which corresponds to a different version of quantum theory.
	Now the no-restriction hypothesis is violated.
	As long as one considers single-system quantum theory only---e.g.~because one models a closed universe in its entirety\footnote{We leave aside the question whether quantum theory can model systems which contain every observer. A reader concerned about this can instead consider a quantum system with more states than are in its environment, so that not every POVM can be dilated to a projective measurement.}---then the choice between POVMs and projective measurements leads to distinct quantum theories that should arguably be carefully distinguished, since they are physically and empirically inequivalent.\footnote{More accurately, they are empirically \emph{conditionally} inequivalent: with POVMs, one \emph{has the freedom} to model observations in a more general way than with projective measurements. The inequivalence appears as soon as one actually has a concrete model of a physical system with at least one observation modelled non-projectively.}
\end{ex}

In the quantum foundations literature, GPTs are frequently considered to be the most general class of conceivable physical theories.
Much research has been dedicated to understanding the scope of this generality, motivating the framework from physical principles, and to finding additional principles which characterize quantum theory with POVMs as a GPT.
But as far as we know, the following question not only remains unanswered, but has apparently not even appeared as a question in print before:

\begin{qstn}
	\label{main_qstn}
	Why are measurements made of effects?
	That is, why should a measurement in a physical theory be given by a tuple of effects as in~\eqref{effect_family} satisfying~\eqref{effect_normalization}?
	Are there meaningful (toy) theories where this is not the case?
\end{qstn}

Our goal in this paper is twofold.
First, we develop a simple mathematical framework in which one can make sense of physical theories where measurements are not necessarily given by tuples of effects.
These \newterm{generalized measurement theories} can be thought of as generalizations of GPTs, with the significant difference that measurements are the primary structure while the concept of state is derived.
Within this framework, \cref{main_qstn} can be investigated formally.
Second, we answer \cref{main_qstn} by showing that in every generalized measurement theory where the measurements are separated by probabilistic states, the measurements are given by normalized tuples of effects (\cref{main_thm}).
As an aside, we also consider the question of what it should mean for a generalized measurement theory to be \emph{classical}, and we provide a characterization theorem for those generalized measurement theories which correspond to tuples of effects with values in a Boolean algebra (\cref{thm_classical}).

\begin{rem}
	\label{scope}
	\begin{enumerate}
		\item As may already have become clear, what we mean by ``measurement'' for the purposes of this paper does \emph{not} involve any specification of what happens to the measured system after the measurement has been performed.
			However, this is of course a very important aspect of measurements in physical theories.
			We leave the inclusion of this aspect in the framework developed here for future work.\footnote{As far as the basic mathematical setting is concerned, the main difference is that one has to work with monads rather than $\Set$-valued functors.}
		\item Throughout the paper, we restrict to the case of measurements with finitely many outcomes, since this is technically simpler but still seems to display all the relevant conceptual structure.
			We do not expect the generalization to infinitely many outcomes (including continuous outcomes) to be meaningfully different, and an analogous rigorous treatment of this case would be possible by working with measurable spaces instead of finite sets.
	\end{enumerate}
\end{rem}

\section{Generalized measurement theories}

So how could we possibly entertain the idea of measurements that are not made of effects?
It may not be obvious how one can even talk about this in mathematical terms.
We solve this problem through abstraction: instead of saying what an individual measurement in a given physical theory is, we instead axiomatize the structure that one has on the collection of \emph{all} measurements on a physical system.
Following e.g.~\cite{FGHL}, one can work with at least two kinds of algebraic structure on the set of measurements:
\begin{itemize}
	\item \emph{Post-processing} by mapping the outcomes to new outcomes, possibly identifying some of them in the process, which is also known as \newterm{coarse-graining}.
		For example, given an $X$-valued measurement $(e_x)_{x \in X}$ in a GPT and any function $f : X \to Y$,\footnote{We do not consider non-deterministic post-processing here, since all non-deterministic post-processings can be obtained by combining deterministic post-processing with mixing (in the sense of the second item).}
		we obtain a $Y$-valued measurement $f_* e$ whose components are
		\begin{equation}
			\label{post_processing}
			(f_* e)_{y \in Y} \coloneqq \sum_{x \in f^{-1}(y)} e_x.
		\end{equation}
	\item \emph{Mixing} any collection of measurements with values in any set $X$, which equips the set of measurements with a convex structure in the form of a barycentric algebra~\cite[Section~12.7]{Schechter}.
		In operational terms, one performs a mixed measurement by randomly selecting a measurement according to some probability distribution on the set of measurements and then performing the selected one.
\end{itemize}

One could now take these two kinds of algebraic structure---with suitable compatibility conditions imposed in the form of equational laws---as the definition of a generalized measurement theory.
However, we have found the structure of post-processing alone to be remarkably powerful, while the structure of mixing does not seem to add much beyond that.
We therefore limit ourselves to post-processing as the only algebraic structure on the collection of measurements.
Axiomatizing this algebraic structure is then easily done using the language of categories and functors.

\begin{defn}
	\label{gmt}
	A \newterm{generalized measurement theory (GMT)} is a functor
	\begin{equation}
		\label{gmt_functor}
		\measfun \: \colon \: \FinSet \to \Set
	\end{equation}
	with $\measfun(1) \cong 1$, and which we think of as assigning:
	\begin{itemize}
		\item to every finite set $X$ the set $\measfun(X)$ of all measurements with outcomes in $X$ on the physical system under consideration,
		\item to every function $f : X \to Y$ the post-processing map $\measfun(f) : \measfun(X) \to \measfun(Y)$.
	\end{itemize}
\end{defn}

\begin{rem}
	\label{gmt_remark}
	The definition deserves some explanation.
	\begin{enumerate}
		\item The symbol ``$1$'' denotes any fixed singleton set.\footnote{This is a customary notation in category theory.}
			Thus the condition $\measfun(1) \cong 1$ states exactly that $\measfun(1)$ is a singleton set again,\footnote{A functor with this property is also called \emph{affine}~\cite{Jacobs_semantics}, although there also is an unrelated notion of ``affine functor'' in algebraic geometry~\cite{NSS}.}
			or equivalently that there must be exactly one measurement with only one outcome.
			As we will see below, this innocuous condition has important consequences for the structure of GMTs.
		\item For a measurement $\alpha \in \measfun(X)$ and a function $f : X \to Y$, we also use the shorthand notation
			\[
				f_* \alpha \coloneqq \measfun(f)(\alpha)
			\]
			for the post-processed measurement.
			Then the requirement that $\measfun$ is a functor means that for all additional $g : Y \to Z$, we must have
			\begin{equation}
				\label{functoriality}
				g_*(f_* \alpha) = (g \circ f)_* \alpha,
			\end{equation}
			and moreover the more trivial requirement that $\id_{X,*} \alpha = \alpha$ for the identity map $\id_X : X \to X$.
			These are the two equational laws that post-processing must satisfy.
		\item\label{topos}
			An idea related to \cref{gmt} is the \emph{spectral presheaf} of topos quantum theory~\cite{Doering,Flori}.
			In our terminology, the spectral presheaf assigns to a quantum system all the projective measurements on that system and keeps track of how those measurements coarse-grain to each other.
			There are important technical differences which lead to the spectral presheaf being defined on a different base category, which depends on the system under consideration.
			Our simple \cref{gmt} avoids this.
		\item We could also allow \emph{random} post-processing by replacing $\FinSet$ with $\FinStoch$, the category of finite sets and stochastic matrices.
			Likewise, we could allow \emph{mixing} by replacing $\Set$ with $\Conv$, a suitable category of convex sets or barycentric algebras and affine maps;
			in this way, every $\measfun(X)$ would come equipped with a notion of mixing that is preserved by the post-processing maps.
			If we were to strictly follow~\cite{FGHL}, we would be implementing both of these changes.
			We have decided against doing so in order to demonstrate that the rudimentary structure of post-processing alone is already surprisingly powerful and expressive.
			In particular, it allows us to forgo the assumption that ordinary numerical probabilities are the correct way to model uncertainty in physical theories.\footnote{We do not find it particularly plausible that this would \emph{not} be the case. Instead, the point is that allowing for the possibility comes at little cost, and it therefore seems worth doing so in order to keep the framework as general and as simple as possible.}

			If random post-processing is allowed, then the following question becomes interesting: for given $\alpha \in \measfun(X)$ and $\beta \in \measfun(Y)$ it becomes an interesting question whether there exists $f : X \to Y$ with $f_* \alpha = \beta$.
			This condition generalizes the \emph{Blackwell order} on statistical experiments~\cite{Blackwell}.
			(If only deterministic post-processing is allowed, then the question still makes sense, but the answer is rarely going to be affirmative because there are too few candidate functions $f$.)
	\end{enumerate}
\end{rem}

In a fixed GPT, taking $\measfun(X)$ to be the set of all $X$-valued measurements as described around~\eqref{effect_family}, and taking $\measfun(f)$ to be the post-processing map~\eqref{post_processing}, defines a GMT.
To verify the functoriality condition~\eqref{functoriality}, we calculate the component of both sides of the equation at any $z \in Z$,
\begin{align*}
	(g_*(f_* \alpha))_z & = \sum_{y \in g^{-1}(z)} (f_* \alpha)_y = \sum_{y \in g^{-1}(z)} \sum_{x \in f^{-1}(y)} \alpha_x, \\[2pt]
	((g \circ f)_* \alpha)_z & = \sum_{x \in (g \circ f)^{-1}(z)} \alpha_x.
\end{align*}
The desired equation now follows because $(g \circ f)^{-1}(z)$ is the disjoint union of the sets $f^{-1}(y)$ for $y \in g^{-1}(z)$.
For example for a fixed Hilbert space $\mathcal{H}$, the GMT of POVMs on $\mathcal{H}$ is given by $\measfun(X)$ being the set of all POVMs with outcomes in $X$ on $\mathcal{H}$, and with $\measfun(f)$ given by post-processing as in~\eqref{post_processing}.

Some further examples of GMTs are as follows.

\begin{ex}[Effect algebras]
	\label{effect_algebras}
	An existing mathematical structure for effects in and beyond GPTs is the formalism of \newterm{effect algebras}~\cite{FB}.
	Since these do not play any further role in this paper, we refer to the literature for the definition.
	Given an effect algebra $E$, we can construct a GMT $\measfun_E$ by defining
	\[
		\measfun_E(X) \coloneqq \left\{ (e_x)_{x \in X} \in E^X \;\bigg|\; \sum_{x \in X} e_x = 1 \right\},
	\]
	and with the action on morphisms again given by post-processing in terms of the addition in $E$.
	If we take $E$ to be the effect algebra of effects in a GPT, then $\measfun_E$ is precisely the measurement functor of that GPT as described above.
	For effect algebras in general, usage of this functor as a mathematical tool goes back to Jacobs~\cite[Definition~20]{Jacobs_convexity} and Staton and Uijlen~\cite{SU}.
	While Jacobs had studied characteristic properties of the functors of the form $\measfun_E : \FinSet \to \Set$, Staton and Uijlen have shown that this construction embeds the category of effect algebras fully faithfully into the functor category $[\FinSet, \Set]$.
\end{ex}

\begin{ex}[Projective measurements]
	For quantum theory with Hilbert space $\mathcal{H}$, one can also consider the GMT of projective measurements.
	This is a subfunctor of the GMT of POVMs.
	It can also be viewed as $\measfun_E$ for the effect algebra $E$ given by the lattice of projections on $\mathcal{H}$.
\end{ex}

\begin{ex}[Classical measurements]
	\label{classical_measurements}
	For a fixed set of ``states'' $S$, there is a GMT which models measurements where the outcomes depend deterministically on the states, so that a measurement with outcomes in $X$ is just a function $S \to X$.
	More concisely,
	\[
		\measfun(X) \coloneqq X^S.
	\]
	The post-processing is given by composition: for $f : X \to Y$ and $\alpha : S \to X$, we put $f_* \alpha \coloneqq f \circ \alpha$.
	It is clear that this implements the idea of a classical GMT.

	More generally, let $\mathcal{B}$ be a Boolean algebra.
	Then this is in particular an effect algebra, and so we can construct a GMT $\measfun_\mathcal{B}$ as above.
	As we will turn to in more detail in~\cref{classical}, such a GMT deserves to be called classical.
	Taking $\mathcal{B} = \{0,1\}^S$ to be the power set of a set $S$ recovers the deterministic GMT of the previous paragraph.
\end{ex}

It could also be of interest to consider probabilistic measurements on a classical state space.
To formalize this as a GMT, let us write $\Delta(X)$ for the set of all probability distributions on a finite set $X$,
\[
	\Delta(X) \coloneqq \left\{ p : X \to [0,1] \;\bigg|\; \sum_{x \in X} p(x) = 1 \right\}.
\]
This is in fact a functor $\Delta : \FinSet \to \Set$, where the functoriality is given by pushforward of probability distributions: for any function $f : X \to Y$,
\begin{equation}
	\label{pushforward}
	(f_* p)(y) \coloneqq \sum_{x \in f^{-1}(y)} p(x).
\end{equation}
Then $\Delta$ is a functor $\FinSet \to \Set$ with $\Delta(1) \cong 1$.

\begin{ex}[Probabilistic measurements and random functions]
	Let again $S$ be a fixed set.
	We can then combine the two previous examples in two different ways.
	\begin{enumerate}
		\item We can consider classical measurements with random outcomes, which corresponds to
			\begin{equation}
				\label{probabilistic_measurements}
				\mathcal{P}(X) \coloneqq \Delta(X)^S.
			\end{equation}
			Like this, a measurement with outcomes in $X$ is a function which assigns to every state in $S$ a probability distribution on $X$.
			The functoriality is obtained by composing the functors $\Delta$ and $(-)^S$ in the obvious way; more explicitly, this means that a measurements gets post-processed by taking the pushforward measure for every state as usual.

			This $\mathcal{P}$ is the natural GMT for a classical system with state space $S$ and probabilistic measurements.
			For $S = 1$, we recover $\mathcal{P} = \Delta$, so that the functor of probability distributions $\Delta$ itself can also be viewed as a (rather trivial) GMT.
		\item On the other hand, we can form a GMT $\mathcal{R}$ by composing the two functors in the opposite order,
			\begin{equation}
				\label{random_functions}
				\mathcal{R}(X) \coloneqq \Delta(X^S).
			\end{equation}
			Like this, a measurement is a ``random function'', i.e.~a probability distribution over maps $S \to X$.
			This is importantly distinct from the previous case, since the randomness is not in the outcomes but already in the choice of which function $S \to X$ to perform.

			As we will see, the GMT $\mathcal{R}$ displays features similar to that of quantum theory, such as contextuality.
			Therefore one probably should not think of $\mathcal{R}$ as a classical GMT.
	\end{enumerate}
\end{ex}

\begin{ex}[Sub-GMTs]
	\label{sub_gmts}
	Given a GMT $\measfun$, let us call a \newterm{sub-GMT} any collection $\measfunb$ of subsets $\measfunb(X) \subseteq \measfun(X)$ that is closed under post-processing: for all functions $f : X \to Y$, we must have
	\[
		f_* \measfunb(X) \,\subseteq\, \measfunb(Y).
	\]
	Then it is clear that $\measfunb$ is again a GMT.
	Let us look at some examples.
	\begin{itemize}
		\item For any fixed Hilbert space $\mathcal{H}$, the projective quantum measurements form a sub-GMT of the GMT of POVMs.
		\item If one considers the functor of probability distributions $\Delta$ as a GMT, then it is also a sub-GMT of the GMT of POVMs on any Hilbert space $\mathcal{H}$, namely the sub-GMT corresponding to those POVMs whose components are all scalar multiples of the identity operator.
		\item If $\mathcal{H}$ is a Hilbert space with a distinguished basis $S$, then the POVMs with diagonal components with respect to this basis form a sub-GMT of the GMT of POVMs on $\mathcal{H}$.
			This sub-GMT is clearly isomorphic to the GMT $\mathcal{P}$ of probabilistic measurements on $S$.
		\item For any set of states $S$, the classical measurements of \cref{classical_measurements} form a sub-GMT of the GMT $\mathcal{R}$ of random functions \eqref{random_functions}.
	\end{itemize}
	
	In general, the intersection of any family of sub-GMTs is again a sub-GMT.
	This implies that for any family of measurements $(\alpha_i \in \measfun(A_i))_{i \in I}$ in any GMT $\measfun$, there is a smallest sub-GMT containing all of these measurements.
	For any $X$, its collection of $X$-valued measurements can be described more concretely as
	\[
		\bigcup_{i \in I} \, \{ f_* \alpha_i \mid f : X_i \to Y \text{ for some } Y \}.
	\]
	In the language of~\cite{FGHL}, this is more or less the \emph{simulation closure} of the family $(\alpha_i)_{i \in I}$, where the difference is that we only close with respect to post-processing and not with respect to mixing.

	For example for any fixed finite set $A$, one can consider the sub-GMT generated by all measurements in $\measfun(A)$.
	When $A$ has two elements, then this is in the spirit of the \emph{effectively dichotomic meters} of~\cite{FGHL}.
	In particular, the GMT of classical measurements of \cref{classical_measurements} has a sub-GMT $\measfun_2 \subseteq \measfun$ where $\measfun_2(X)$ consists of all functions $S \to X$ whose image has cardinality $\leq 2$.
	In a GMT like this, all information that one can obtain by measuring is effectively binary.
\end{ex}

As especially \cref{random_functions} might illustrate, a GMT on its own does not yet have much physical content.
In particular, there is no notion of state yet, and relatedly no way in which measurements actually \emph{have} outcomes.
It is therefore pertinent to introduce a notion of state.
But before doing so, we first take a closer look at what extra structure a GMT automatically comes equipped with.

\begin{lem}
	\label{delta}
	Every GMT $\measfun$ comes equipped with canonical maps
	\[
		\delta_X \: : \: X \longrightarrow \measfun(X)
	\]
	which are natural in $X$ in the sense that for all $f : X \to Y$, the diagram
	\begin{equation}
		\label{delta_natural}
		\begin{tikzcd}
			X \arrow[r, "\delta_X"] \arrow[d, "f"'] & \measfun(X) \arrow[d, "\measfun(f)"] \\
			Y \arrow[r, "\delta_Y"'] & \measfun(Y)
		\end{tikzcd}
	\end{equation}
	commutes.
\end{lem}

\begin{proof}
	We identify an element $x \in X$ with the function $1 \to X$ which sends the unique element of $1$ to $x$.
	Then we can take
	\[
		\delta_X(x) \coloneqq x_* \tau,
	\]
	where $\tau \in \measfun(1)$ is the unique measurement with only one outcome.
	The commutativity of~\eqref{delta_natural} follows straightforwardly from the functoriality~\eqref{functoriality}.
\end{proof}

In the language of category theory, \cref{delta} says that there is a canonical natural transformation $\delta : \id_{\FinSet} \Rightarrow \measfun$ from the identity functor on $\FinSet$ to the measurement functor $\measfun$.
The physical interpretation is that $\delta_X$ assigns to every outcome $x \in X$ the trivial measurement which always returns $x$.

In fact, \cref{delta} is a special case of a more general phenomenon~\cite[Remark~4.10]{AMSW}: every functor between categories of sets comes equipped with a unique \emph{strength} in the sense of a canonical natural transformation with components
\[
	s_{X,Y} \: : \: X \times \measfun(Y) \longrightarrow \measfun(X \times Y)
\]
satisfying certain unit and associativity conditions~\cite[(1.7)--(1.8)]{Kock}.
Concretely, $s_{X,Y}$ maps an outcome $x \in X$ and a measurement $\alpha \in \measfun(Y)$ to the measurement which first performs $\alpha$ and then returns an element of $X \times Y$ which is $x$ in the first component and the outcome of $\alpha$ in the second component.
More formally speaking, we have $s_{X,Y}(x, \alpha) = \iota_* \alpha$, where
\begin{align*}
	\iota \: : \: Y & \longrightarrow X \times Y \\
	y & \longmapsto (x, y).
\end{align*}
Modulo trivial isomorphisms like $X \times 1 \cong X$, taking $Y = 1$ and using $\measfun(1) = \{\tau\}$ recovers $\delta_X$ as the strength component $s_{X,1}$.

Another automatic piece of structure is \newterm{marginalization}: for any two finite sets $X$ and $Y$, there is a canonical map
\begin{align}
	\label{marginalization}
	\begin{split}
		\measfun(X \times Y) & \longrightarrow \measfun(X) \times \measfun(Y) \\
		\alpha & \longmapsto (\pi_{X,*} \alpha, \pi_{Y,*} \alpha),
	\end{split}
\end{align}
where $\pi_X : X \times Y \to X$ and $\pi_Y : X \times Y \to Y$ are the projection maps.
So every measurement with values in $X \times Y$ ``marginalizes'' to a pair of measurements with values in $X$ and $Y$.
For $\measfun = \Delta$ the functor of probability distributions, this is the usual marginalization of joint distributions to their marginals.
In general, these maps are again automatically natural and satisfy certain unit and associativity conditions~\cite{FP}, making $\measfun$ into an \emph{oplax monoidal functor}~\cite[Definition~3.2]{AM}.

\section{States in generalized measurement theories}

We now turn to the question of how to make the measurements in a GMT actually \emph{have} outcomes.
This question is addressed by various notions of \emph{state}, which we define as mappings from measurements to outcomes satisfying certain consistency conditions.

\begin{defn}
	\label{deterministic_state}
	Given a GMT $\measfun$, a \newterm{deterministic state} $s$ is a family of maps
	\[
		s_X \: : \: \measfun(X) \longrightarrow X
	\]
	such that for all $f : X \to Y$, the diagram
	\[
		\begin{tikzcd}
			\measfun(X) \arrow[r, "s_X"] \arrow[d, "\measfun(f)"'] & X \arrow[d, "f"] \\
			\measfun(Y) \arrow[r, "s_Y"'] & Y
		\end{tikzcd}
	\]
	commutes.
\end{defn}

More succinctly, a deterministic state is a natural transformation $s : \measfun \Rightarrow \id_{\FinSet}$.
To look at a first example, recall the classical GMT of \cref{classical_measurements} where $\measfun(X) = X^S$, where $S$ is a fixed set of ``classical states''.
Then the following result justifies usage of the word ``state'' for both the elements of $S$ and the natural transformations $s : \measfun \Rightarrow \id_{\FinSet}$.

\begin{lem}
	\label{classical_states}
	The deterministic states of this classical GMT are in bijection with the elements of $S$.
\end{lem}

\begin{proof}
	Since the functor under consideration is a hom-functor $\FinSet(S, -)$, this is an instance of the Yoneda lemma~\cite[Theorem~2.2.4]{Riehl}.
\end{proof}

Many GMTs do not have any deterministic states at all.
This phenomenon is known as \newterm{state-independent contextuality}.
Indeed both state-independent contextuality and our \cref{deterministic_state} are concerned with assigning an outcome to every measurement in a way that is consistent under deterministic post-processing, and in particular under coarse-graining.

\begin{thm}[Kochen--Specker]
	\label{KS}
	The GMTs of
	\begin{enumerate}
		\item projective measurements on a quantum system with Hilbert space of dimension $\geq 3$,
		\item POVMs on a quantum system with Hilbert space of dimension $\geq 2$,
	\end{enumerate}
	have no deterministic states.
\end{thm}

When one proves the Kochen--Specker theorem---say in the first form---one typically finds a \emph{finite} family of projective measurements, determined by a configuration of rays in the Hilbert space, which already establishes the contradiction with a deterministic assignment of outcomes.
In our setting, this amounts to finding a sub-GMT generated by finitely many measurements in the sense of \cref{sub_gmts}, such that already this sub-GMT does not have any deterministic state.\footnote{Note that this formulation bears some resemblance to the ``sheaf-theoretic'' approach to contextuality~\cite{AB,SU}.}

\begin{prop}
	\label{random_functions_contextuality}
	For any set of states $S \neq \emptyset$, the GMT of random functions $\mathcal{R}$ \cref{random_functions} does not have any deterministic states.
\end{prop}

\begin{proof}
	Consider $X \coloneqq \{0,1\}$ and the constant functions $c_0, c_1 : S \to X$ mapping all elements of $S$ to $0$ and $1$, respectively.
	Then
	\[
		\alpha \coloneqq \frac{[c_0] + [c_1]}{2}
	\]
	is an element of $\mathcal{R}(\{0,1\})$ which satisfies the invariance condition $\alpha = \mathsf{swap}_* \alpha$, where $\mathsf{swap} : \{0,1\} \to \{0,1\}$ is the function which swaps $0$ and $1$.
	If we now had a deterministic state $s : \mathcal{R} \Rightarrow \id_{\FinSet}$, then the commutativity of the diagram in \cref{deterministic_state} for $f = \mathsf{swap}$ would imply that $s_{\{0,1\}}(\alpha) = \mathsf{swap}(s_{\{0,1\}}(\alpha))$.
	But since no element of $\{0,1\}$ is fixed by $\mathsf{swap}$, this is a contradiction.
\end{proof}

Of course, in quantum theory we have gotten used to the idea that the concept of deterministic state is too limited, as the Born rule is fundamentally probabilistic.
Considering probabilistic states in GMTs is possible through a simple modification of \cref{deterministic_state} which replaces the sets of outcomes by the sets of probability distributions over outcomes.

\begin{defn}
	\label{probabilistic_state}
	Given a GMT $\measfun$, a \newterm{probabilistic state} $\rho$ is a family of maps
	\[
		\rho_X \: : \: \measfun(X) \longrightarrow \Delta(X)
	\]
	such that for all $f : X \to Y$, the diagram
	\begin{equation}
		\label{probabilistic_state_diagram}
		\begin{tikzcd}
			\measfun(X) \arrow[r, "\rho_X"] \arrow[d, "\measfun(f)"'] & \Delta(X) \arrow[d, "\Delta(f)"] \\
			\measfun(Y) \arrow[r, "\rho_Y"'] & \Delta(Y)
		\end{tikzcd}
	\end{equation}
	commutes.
\end{defn}
	
We could also consider \emph{possibilistic} states by replacing $\Delta(X)$ with the nonempty power set considered below in~\eqref{possibilistic_measurements}.
While we will not pursue this further in the present paper, we find it intriguing that the GMT framework allows for the consideration of different types of uncertainty in a unified way.

Similar to what happened with deterministic states and contextuality, probabilistic states are a well-known concept in quantum theory:

\begin{thm}[Gleason]
	\label{Gleason}
	In the GMTs of
	\begin{enumerate}
		\item\label{gleason1}
			projective measurements on a quantum system with Hilbert space of finite dimension $\geq 3$,
		\item POVMs on a quantum system with Hilbert space of finite dimension $\geq 2$,
	\end{enumerate}
	the probabilistic states are in bijection with the mixed quantum states (density matrices).
\end{thm}

\begin{rem}
	It is worth noting that Gleason's original proof of~\ref{gleason1} already applies to separable infinite-dimensional Hilbert spaces.
	However, this requires an assumption of countable additivity.
	In our setting, formulating this would require allowing measurements with countable infinitely many outcomes.
	We expect this generalization to be straightforward, and that a version of \cref{Gleason} will then obtain.
	See also \cref{scope}.
\end{rem}

\begin{ex}[Unknown functions]
	\label{unknown_functions}
	We construct a GMT $\mathcal{U}$ similar to the GMT of random functions from \cref{random_functions}, but where the uncertainty about the functions is \emph{possibilistic} rather than probabilistic.

	First, we can formalize possibilistic uncertainty in terms of a functor $\mathcal{P} : \FinSet \to \Set$ analogous to the functor $\Delta$ of probability distributions.
	If we only keep track of which outcomes are possible rather than how likely they are, then we are dealing with nonempty sets of outcomes rather than probability distributions, so that we take $\mathcal{P}(X)$ to be the set of all nonempty subsets of $X$,
	\begin{equation}
		\label{possibilistic_measurements}
		\mathcal{P}(X) \coloneqq 2^X \setminus \{\emptyset\}.
	\end{equation}
	The functoriality on $f : X \to Y$ is given by the formation of the image subset, that is
	\[
		f_* A \coloneqq f(A)
	\]
	for all $A \in \mathcal{P}(X)$, in accordance with the idea that an outcome $y \in Y$ is possible for the post-processed measurement if and only if there is some $x \in f^{-1}(y)$ which is possible for the original measurement.

	As a variant of \cref{random_functions}, we hence obtain a GMT $\mathcal{U}$ of ``unknown functions'' given by
	\[
		\mathcal{U}(X) \coloneqq \mathcal{P}(X^S),
	\]
	and with post-processing defined in the obvious way.
	Again this seems to be meaningfully interpretable as a GMT: a measurement implements a function from $S$ to $X$, but there can be uncertainty about which function this is.
\end{ex}

The example of unknown functions is interesting because it displays the following even more dramatic form of contextuality.

\begin{prop}
	\label{unknown_functions_contextuality}
	For any set of states $S \neq \emptyset$, the GMT of unknown functions $\mathcal{U}$ does not have any probabilistic state.
\end{prop}

\begin{proof}
	We argue similarly as in the proof of \cref{random_functions_contextuality}.
	Indeed let $\alpha \in \mathcal{U}(\{0,1,2\})$ be the set of \emph{all} functions $S \to \{0,1,2\}$, and similarly for $\beta \in \mathcal{U}(\{0,1\})$.
	Then we have $\beta = f_* \alpha$ for every surjective function $f : \{0,1,2\} \to \{0,1\}$.

	Now a probabilistic state would have to assign to $\alpha$ a probability distribution $\rho_{\{0,1,2\}}(\alpha)$ on $\{0,1,2\}$ and to $\beta$ a probability distribution $\rho_{\{0,1\}}(\beta)$ on $\{0,1\}$ such that for every surjective $f : \{0,1,2\} \to \{0,1\}$, we likewise have
	\[
		\rho_{\{0,1\}}(\beta) = f_* \rho_{\{0,1,2\}}(\alpha).
	\]
	This equation is a linear constraint on the outcome probabilities in the distributions $\rho_{\{0,1,2\}}(\alpha)$ and $\rho_{\{0,1\}}(\beta)$.
	Upon considering all six choices of $f$, one easily sees that no solution to these constraints exists.
\end{proof}

\section{Measurements as tuples of effects}

We now turn to the main question of the paper: why are measurements made of effects?
The first important observation we will make is that that this is not always the case in a rather strong sense.
In order to do so, it is helpful to have a criterion which can distinguish GMTs that come from effect algebras from those that do not.
To this end, recall that an effect is also often defined as a binary measurement.
Furthermore, if a measurement $\alpha \in \measfun(X)$ is specified by a tuple of effects $(\alpha_x)_{x \in X}$, then its component $\alpha_x$ at an $x \in X$ can be reconstructed by post-processing with the indicator function $\chi_x : X \to [0,1]$ of $x$.
Thus the GMTs which arise from effect algebras have the following special property.\footnote{See also Jacobs' notion of \emph{effectus}, which is based on this observation and additional properties of measurements arising from effect algebras~\cite{Jacobs_newdirections}.}

\begin{defn}
	A GMT $\measfun$ is \newterm{binarizable} if for any finite set $X$ and distinct $\alpha, \beta \in \measfun(X)$, there is $f : X \to \{0,1\}$ such that $f_* \alpha \neq f_* \beta$.
\end{defn}

In fact, binarizability is a necessary condition for even being \emph{embeddable} into a GMT defined in terms of effects.
In other words, even if one drops the no-restriction hypothesis as in~\cite{FGHL} from an ``effectful'' GMT, then binarizability still obtains.

\begin{prop}
	\label{non_effect_gmt}
	There are GMTs $\measfun$ for which there is no sub-GMT embedding $\measfun \subseteq \measfun_E$ for any effect algebra $E$.
\end{prop}

\begin{proof}
	We first construct a non-binarizable functor $\measfun : \FinSet \to \Set$ and then describe its interpretation as a GMT.
	For a finite set $X$, we write
	\[
		\binom{X}{3}
	\]
	for the set of $3$-element subsets of $X$.
	Then we define
	\[
		\measfun(X) \coloneqq \{\tau_X \} + \binom{X}{3},
	\]
	where we think of the element $\tau_X$ as a trivial (completely uninformative) measurement.
	For every function $f : X \to Y$, we define $f_* \tau \coloneqq \tau$, and take $f_* A$ for a three-element subset $A \subseteq X$ to be given by the image $f(A) \subseteq Y$ if it still has three elements, and $f_* A \coloneqq \tau$ otherwise.
	The functoriality property~\eqref{functoriality} is straightforward to verify based on the fact that taking images of subsets cannot increase their cardinality.
	Now since $\measfun(\{0,1\})$ contains only one element, it is obvious that binarizability is violated, and therefore $\measfun$ cannot be embedded into any GMT arising from an effect algebra.

	As far as the physical interpretation is concerned, we think of this $\measfun$ as modelling a physical system on which exactly one genuinely ternary measurement can be performed.
	The unusual feature of this measurement is that coarse-graining it to binary already makes it completely uninformative.
\end{proof}

So arguably a better question is, \emph{when} are measurements made of effects?
And is there a good explanation for why this is the case in quantum theory?

In order to answer this question, we undertake a more systematic comparison of GPTs and GMTs.
Perhaps unsurprisingly, the set of probabilistic states of a GMT is essentially a GPT.
To formalize this idea, we now do need a suitably general notion of GPT which allows for an infinite-dimensional state space.

\begin{defn}[{e.g.~\cite{Ludwig_axiomatic}}]
	\label{gpt}
	A \newterm{general probabilistic theory (GPT)} consists of a locally convex vector space $V$ together with a compact convex subset $\mathcal{S} \subseteq V$ and a distinguished functional $\tr \in V^*$ such that $\tr(s) = 1$ for all $s \in \mathcal{S}$.
\end{defn}

Now given a GMT $\measfun$, we can take convex combinations of probabilistic states in order to obtain new probabilistic states.
One way to build a vector space $V$ in which these convex combinations live is to consider the set of natural transformations
\begin{equation}
	\label{quasiprobabilistic_states}
	\measfun \Longrightarrow \tilde{\Delta},
\end{equation}
where $\tilde{\Delta} : \FinSet \to \Set$ is the functor which maps a finite set $X$ to the set of tuples of real numbers,
\[
	\tilde{\Delta}(X) \coloneqq \mathbb{R}^X,
\]
with functoriality given by pushforward as in~\eqref{pushforward}.
The idea here is to obtain a vector space of quasi-probabilistic states where one drops both the nonnegativity and the normalization of probability.
Now these quasi-probabilistic states~\eqref{quasiprobabilistic_states} clearly form a vector space $V$ with respect to pointwise addition and scalar multiplication.
We make it into a locally convex space by equipping it with the weakest topology which makes the projection maps to the $\R^X$ continuous.

Now the probabilistic states form a compact convex subset $\mathcal{S} \subseteq V$, where the compactness holds by Tychonoff's theorem.\footnote{By definition, $\mathcal{S}$ is a closed subset of the product of compact sets $\prod_{X \in \FinSet, \alpha \in \measfun(X)} \Delta(X)$, and hence compact itself.}
Furthermore, the normalization functional is given by
\[
	\tr(\rho) \coloneqq \rho_1(\tau)_\ast,
\]
where we write $1 = \{\ast\}$ for the singleton set.
Now for every measurement $\alpha \in \measfun(X)$ and every $x \in X$, we have an evaluation map
\begin{align*}
	\hat{\alpha}_x \: : \: V & \longrightarrow \R \\
	\rho & \longmapsto \rho_X(\alpha)_x,
\end{align*}
and it is easy to see that this is an effect.
Summing over $x$ gives
\[
	\sum_{x \in X} \hat{\alpha}_x(\rho) = \sum_{x \in X} \rho_X(\alpha)_x = \rho_1(\tau)_\ast = \tr(\rho),
\]
where the second equation holds by the naturality~\eqref{probabilistic_state_diagram} of $\rho$ applied to the unique map $X \to 1$.
Therefore the normalization condition~\eqref{effect_normalization} is satisfied, and we have associated to the measurement $\alpha$ in the GMT $\measfun$ a measurement $(\hat{\alpha}_x)_{x \in X}$ in the GPT of probabilistic states of $\measfun$.

In many cases, this construction is an embedding of the original GMT into the GMT associated to the GPT.
The following definition is designed to capture when this is the case.

\begin{defn}
	A GMT is \newterm{probabilistically separated} if every two distinct measurements differ on some probabilistic state.
\end{defn}

By the above considerations, we have mostly already proven the following key point.

\begin{thm}
	\label{main_thm}
	Every probabilistically separated GMT is a sub-GMT of a GPT, and in particular measurements in such a GMT can be identified with tuples of effects in the GPT (possibly with additional restrictions).
\end{thm}

\begin{proof}
	If $\measfun$ is probabilistically separated, then the map $\alpha \mapsto (\hat{\alpha}_x)_{x \in X}$ described above is injective for every $X$.
	It remains to be shown that this construction respects post-processing, or in other words that the GMT $\measfun$ actually embeds into the GMT of effects of the associated GPT.
	Formally, this means that for $f : X \to Y$ and $\alpha \in \measfun(X)$, we need to show that
	\[
		\widehat{f_* \alpha}_y = \sum_{x \in f^{-1}(y)} \hat{\alpha}_x.
	\]
	It is enough to show that this holds upon evaluation on any probabilistic state $\rho$, in which case we get
	\[
		\widehat{f_* \alpha}_y(\rho) = \rho_Y(f_* \alpha)_y = \sum_{x \in f^{-1}(y)} \rho_X(\alpha)_x = \sum_{x \in f^{-1}(y)} \hat{\alpha}_x(\rho),
	\]
	as was to be shown.
	Here, the second equation holds by the naturality~\eqref{probabilistic_state_diagram} of $\rho$ applied to $f : X \to Y$, while the others are by unfolding definitions.
\end{proof}

\begin{rem}
	We view \cref{main_thm} as a candidate conceptual explanation for why measurements in quantum theory---and in GPTs more generally---are given by tuples of effects.
	But in order for this to be a satisfactory explanation, one needs to explain why measurements in a physical theory should be probabilistically separated.
	Assuming that probability distributions are the correct way to model uncertainty, we believe that at least a partial answer to this latter question can be given as follows: 
	if two measurements cannot be distinguished by any probabilistic state, then they are indistinguishable in any operational sense, and so it is natural to look for a mathematical formalism for the physical theory in which the two measurements are modelled by the same mathematical object.
	And the way in which we treat measurements in quantum theory and in GPTs does treat them \emph{not} as physical protocols or descriptions of apparatuses---of which there can be many different ones which always yield the same measurement statistics---but rather as mathematical objects in a more minimal mathematical formalism which has already undertaken a quotienting by operational indistinguishability.\footnote{Compare with the quotienting which happens in Spekkens' approach to contextuality~\cite{Spekkens_contextuality} as well as his discussion of the identity of indiscernibles~\cite{Spekkens_leibniz}.}
\end{rem}

\section{Compatibility and classicality}
\label{classical}

Finally, we turn to a discussion of when a GMT should be thought of as classical.
Let us start by describing what it means for two measurements to be compatible in a GMT.
As we will see, this question is already poses a number of interesting subtleties, as there are different degrees of compatibility.
Perhaps the most naive definition of compatibility is the following.

\begin{defn}
	\label{weak_compatibility}
	Given a GMT $\measfun$, a finite family of measurements $(\alpha_i \in \measfun(X_i))_{i=1}^n$ is \newterm{weakly compatible} if there exists a finite set $Y$ and functions
	
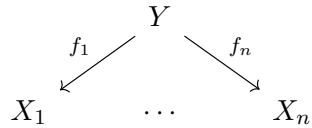
\[
		\begin{tikzcd}
			& Y \arrow{dl}[swap]{f_1} \arrow{dr}{f_n} & \\
			X_1 & \cdots & X_n
		\end{tikzcd}
	\]
	as well as a measurement $\beta \in \measfun(Y)$ such that
	\[
		\alpha_i = f_{i,*} \beta \qquad \forall i = 1, \ldots, n.
	\]
\end{defn}

By the universal property of categorical products in $\FinSet$, one can restrict to the case
\[
	Y = X_1 \times \cdots \times X_n
\]
and $f_i$ the projection maps without loss of generality.
For example, for the GMT of POVMs in quantum theory, it is known that weak compatibility in this sense is not implied by pairwise weak compatibility~\cite{KHF}.

The following is then the first natural definition of classicality that one might give in the GMT framework.

\begin{defn}
	A GMT $\measfun$ is \newterm{weakly classical} if any finitely many measurements are compatible.
\end{defn}


\begin{ex}
	The weird GMT from the proof of \cref{non_effect_gmt} is weakly classical.

	Indeed given any finite family $(\alpha_i \in \measfun(X_i))_{i=1}^n$, we can take $Y$ to be any three-element set and $\beta \coloneqq Y$ itself.
	Then for those $i$ for which $\alpha_i$ is trivial, we can take $f_i : Y \to X_i$ to be any constant function, 
	while for those $i$ for which $\alpha_i$ is a three-element subset of $X_i$, we can take $f_i : Y \to X_i$ to be any map with $f_i(Y) = \alpha_i$.
\end{ex}

It seems questionable whether the previous example of a GMT should really be considered as classical.
As we will see, the following stronger notion of compatibility excludes it, and in fact defines a class of GMTs which arguably clearly deserve the name ``classical''.

\begin{defn}
	\label{strong_compatibility}
	Given a GMT $\measfun$, a finite family of measurements $(\alpha_i \in \measfun(X_i))_{i=1}^n$ is \newterm{strongly compatible} if weak compatibility holds as in \cref{weak_compatibility} with the following additional property:
	For any finite set $Z$ with $\gamma \in \measfun(Z)$ and maps $g_i : Z \to X_i$ such that
	\[
		g_{i,*} \gamma = \alpha_i \qquad \forall i = 1, \ldots, n,
	\]
	there is a unique function $h : Z \to Y$ such that
	\[
		g_i = f_i \circ h \qquad \forall i = 1, \ldots, n
	\]
	and
	\[
		\beta = h_* \gamma.
	\]
\end{defn}

In other words, there should be a \emph{universal} measurement $\beta$ which post-processes to the $\alpha_i$, in the sense that every other measurement $\gamma$ which post-processes to the $\alpha_i$ already post-processes to $\beta$ in a unique way.
In categorical terms, such a universal $\beta$ is a product of the $\alpha_i$ in the category of elements of $\measfun$~\cite[Section~2.4]{Riehl}.
Similarly as in the weak case, we have the following more workable characterization.

\begin{lem}
	\label{strong_compatibility_product}
	A family of measurements $(\alpha_i \in \measfun(X_i))_{i=1}^n$ is strongly compatible if and only if there is a unique
	\[
		\beta \in \measfun(X_1 \times \cdots \times X_n)
	\]
	such that $\alpha_i = \pi_{i,*} \beta$ for all $i = 1, \ldots, n$, where $\pi_i : X_1 \times \cdots \times X_n \to X_i$ is the projection map.
\end{lem}

\begin{proof}
	We start with the ``if'' direction.
	Since the cartesian product $X_1 \times \cdots \times X_n$ is the categorical product in $\FinSet$, we obtain the unique
	\[
		h \: : \: Z \longrightarrow X_1 \times \cdots \times X_n
	\]
	from its universal property as the tupling $h = (g_1, \ldots, g_n)$.
	Now all conditions are already clear except for $\beta = h_* \gamma$, which follows from the fact that
	$\pi_{i,*} h_* \gamma = g_{i,*} \gamma = \alpha_i$ and the uniqueness of $\beta$ as the only element of $\measfun(X_1 \times \cdots \times X_n)$ with $\pi_{i,*} \beta = \alpha_i$ for all $i$.

	For the ``only if'' direction, suppose that the strong compatibility condition holds for some $Y$, $f_i$ and $\beta$ as in the definition.
	Then we can consider the tupling $(f_1, \dots, f_n) : Z \to X_1 \times \cdots \times X_n$ and take
	\[
		\eta \coloneqq (f_1, \ldots, f_n)_* \beta \,\in\, \measfun(X_1 \times \cdots \times X_n).
	\]
	This satisfies
	\[
		\pi_{i,*} \eta = f_{i,*} \beta = \alpha_i
	\]
	as required.
	To prove that it is the unique such measurement, suppose that $\eta' \in \measfun(X_1 \times \dots \times X_n)$ also satisfies $\pi_{i,*} \eta' = \alpha_i$ for all $i$.
	Then we can take $Z = X_1 \times \dots \times X_n$ and $g_i \coloneqq \pi_i$ in the definition of strong compatibility, and we obtain a unique $h : X_1 \times \dots \times X_n \to Y$ such that 
	\[
		\pi_i = f_i \circ h \qquad \forall i = 1, \ldots, n
	\]
	and $\beta = h_* \eta'$.
	Now we have $(f_1, \ldots, f_n) \circ h = (\pi_1, \ldots, \pi_n) = \id_{X_1 \times \cdots \times X_n}$ by the universal property of the product, and therefore
	\[
		\eta' = (f_1, \ldots, f_n)_* h_* \eta' = (f_1, \ldots, f_n)_* \beta = \eta,
	\]
	which proves the uniqueness of $\eta$.
\end{proof}

\begin{defn}
	A GMT $\measfun$ is \newterm{strongly classical} if any finitely many (or equivalently any two) measurements are strongly compatible.
\end{defn}

Indeed the fact that the existence binary products in a category implies the existence of finite products shows that pairwise strong compatibility is enough.
This is a nice feature of the definition which distinguishes it from weak compatibility.

\begin{cor}
	\label{strong_classicality_marginalization}
	A GMT $\measfun$ is strongly classical if and only if the marginalization maps~\eqref{marginalization} are bijections for all finite sets $X$ and $Y$.
\end{cor}

\begin{proof}
	By \cref{strong_compatibility_product} applied to all pairs of elements of $\measfun(X)$ and $\measfun(Y)$.
\end{proof}

Informally, \cref{strong_classicality_marginalization} states that strong classicality is equivalent to the condition that the measurements with outcomes in $X \times Y$ are precisely the pairs of measurements with outcomes in $X$ and $Y$.

The facts that strong compatibility corresponds to the existence of a categorical product in the category of elements, and that a categorical product is a certain kind of limit, suggest looking at other kinds of limits as well.
Inspired by this, the following definition postulates the existence of \emph{equalizers} in the category of elements.
Our choice of terminology will be explained by \cref{projective_projective}.

\begin{defn}
	A GMT $\measfun$ is \newterm{projective} if for any finite sets $X$ and $Y$ and any $\alpha \in \measfun(X)$ and functions
	\[
		\begin{tikzcd}
			X \ar[r, "f", shift left] \ar[r, "g"', shift right] & Y
		\end{tikzcd}
	\]
	with $f_* \alpha = g_* \alpha$, there exist a finite set $E$ and a measurement $\sigma \in \measfun(E)$ and a function $e : X \to E$ such that:
	\begin{enumerate}
		\item $f \circ e = g \circ e$ and $e_* \sigma = \alpha$.
		\item\label{projective_universality}
			For any finite set $E'$ and $\sigma' \in \measfun(E')$ and $e' : X \to E'$ such that $f \circ e' = g \circ e'$ and $e'_* \sigma' = \alpha$, there is a unique function $h : E' \to E$ such that $e' = e \circ h$ and $\sigma = h_* \sigma'$.
	\end{enumerate}
\end{defn}

The following reformulation, which is analogous to \cref{strong_compatibility_product}, makes this definition also more tangible and easier to check.

\begin{lem}
	\label{projective_equalizer}
	A GMT $\measfun$ is projective if and only if for any finite sets $X$ and $Y$ and any $\alpha \in \measfun(X)$ and functions
	\[
		\begin{tikzcd}
			X \ar[r, "f", shift left] \ar[r, "g"', shift right] & Y
		\end{tikzcd}
	\]
	with $f_* \alpha = g_* \alpha$, the measurement $\alpha$ is supported on the subset
	\[
		S \coloneqq \{x \in X \mid f(x) = g(x)\}
	\]
	in the sense that there exists a unique measurement $\sigma \in \measfun(S)$ with $e_* \sigma = \alpha$ for the inclusion map $i : S \to X$.
\end{lem}

Also the proof is largely parallel to the proof of \cref{strong_compatibility_product}.

\begin{proof}
	We start with the ``if'' direction.
	Since $i$ is the equalizer of $f$ and $g$ in $\FinSet$ by construction, it is enough to show that $\sigma = h_* \sigma'$ while all other conditions are already clear.
	To this end, consider any $E'$ and $\sigma'$ and $e' : X \to E'$ such that $f \circ e' = g \circ e'$ and $e'_* \sigma' = \alpha$.
	Then the universal property of the equalizer $i : S \to X$ implies that there is a unique $h : E' \to S$ such that $e' = i \circ h$. 
	The equation $\sigma = h_* \sigma'$ holds by $i_* h_* \sigma' = (i \circ h)_* \sigma' = e'_* \sigma' = \alpha$ and the uniqueness of $\sigma$.

	For the ``only if'' direction, suppose that the projectivity condition holds for some $E$ and $\sigma$ and $e$ as in the definition.
	Then we consider the map $k : E \to S$ which acts like $e$ but with codomain $S$ instead of $E$, meaning that $e = i \circ k$.
	We argue that
	\[
		\nu \coloneqq k_* \sigma
	\]
	is the unique element of $\measfun(S)$ with $i_* \nu = \alpha$.
	Indeed
	\[
		i_* \nu = i_* k_* \sigma = e_* \sigma = \alpha,
	\]
	as required.
	To prove that it is unique with this property, suppose that $\nu' \in \measfun(S)$ also satisfies $i_* \nu' = \alpha$.
	Then we can take $E' = S$ and $\sigma' = \nu'$ and $e' = k$ in the definition of projectivity, and we obtain a unique $h : S \to E$ such that $k = e \circ h$ and $\sigma = h_* \nu'$.
	Now we have $k \circ h = \id_S$ by the universal property of the equalizer, and therefore
	\[
		\nu' = k_* h_* \nu' = k_* \sigma = \nu,
	\]
	which proves the uniqueness of $\nu$.
\end{proof}

In categorical language, \cref{projective_equalizer} can be formulated as the following analogue of \cref{strong_classicality_marginalization}.

\begin{cor}
	A GMT $\measfun$ is projective if and only if $\measfun : \FinSet \to \Set$ preserves equalizers.
\end{cor}

\begin{proof}
	This is a reformulation of \cref{projective_equalizer} based on the fact that $S$ is the equalizer of $f$ and $g$ in $\FinSet$ by construction.
\end{proof}

\begin{ex}
	\label{projective_projective}
	In quantum theory with Hilbert space $\mathcal{H}$, the GMT of projective measurements is projective, while the GMT of POVMs is not.

	For the former, suppose that $(\alpha_x)_{x \in X}$ is a projective measurement and that $f, g : X \to Y$ are such that $f_* \alpha = g_* \alpha$.
	This means that for all $y \in Y$, we have
	\begin{equation}
		\label{projective_condition}
		\sum_{x \in f^{-1}(y)} \alpha_x = \sum_{x \in g^{-1}(y)} \alpha_x.
	\end{equation}
	Then since the $\alpha_x$ are projection operators with $\sum_{x \in X} \alpha_x = 1$, they are \emph{all} pairwise orthogonal, 
	and thus two partial sums can be equal only if they contain precisely the same nonzero terms.
	Therefore we must have $\alpha_x = 0$ for all $x$ with $f(x) \neq g(x)$, since otherwise $\alpha_x$ would contribute nontrivially to the left-hand side of~\eqref{projective_condition} with $y = f(x)$ but not to the right-hand side.
	In other words, $\alpha$ must be supported on the subset
	\[
		S \coloneqq \{x \in X \mid f(x) = g(x)\},
	\]
	as claimed.

	To see that the GMT of POVMs is not projective, we can argue by symmetries as in the proofs of \cref{random_functions_contextuality} and \cref{unknown_functions_contextuality}.
	Indeed by taking both POVM elements to be half the identity operator, we obtain $\alpha \in \measfun(\{0,1\})$ which is invariant under $\mathsf{swap} : \{0,1\} \to \{0,1\}$.
	Then taking $f \coloneqq \mathsf{swap}$ and $g \coloneqq \id_{\{0,1\}}$ gives $f_* \alpha = g_* \alpha$.
	But now the subset $S$ is empty, and there is no POVM on the empty set, so that the condition of \cref{projective_equalizer} fails.
\end{ex}

The following is the main result of this section.
It characterizes the GMTs associated to Boolean algebras (\cref{classical_measurements}) in terms of the special properties we have considered.
To state it, let us say that a GMT $\measfun$ is \newterm{nontrivial} if $|\measfun(X)| > 1$ for some nonempty set $X$.

\begin{thm}
	\label{thm_classical}
	For every nontrivial GMT, the following are equivalent:
	\begin{enumerate}
		\item\label{main_thm_onlyclassical}
			$\measfun$ is strongly classical.
		\item\label{main_thm_classical}
			$\measfun$ is strongly classical and projective.
		\item\label{main_thm_boolean}
			$\measfun$ is isomorphic to a GMT of the form $\measfun_\mathcal{B}$ for a Boolean algebra $\mathcal{B}$.
	\end{enumerate}
\end{thm}

\begin{proof}
	We first prove the equivalence of~\ref{main_thm_classical} and~\ref{main_thm_boolean}.
	From~\ref{main_thm_boolean} to~\ref{main_thm_classical}, we need to show that every GMT of the form $\measfun_\mathcal{B}$ for a Boolean algebra $\mathcal{B}$ is strongly classical and projective.
	For strong classicality, we use the criterion of \cref{strong_compatibility_product}.
	So let $(\alpha_i \in \measfun_\mathcal{B}(X_i))_{i \in I}$ be given.
	Then we can take $\beta \in \measfun_\mathcal{B}(X_1 \times \dots \times X_n)$ to be given by
	\[
		\beta_{(x_1, \dots, x_n)} \coloneqq \alpha_{x_1} \wedge \dots \wedge \alpha_{x_n}
	\]
	for all tuples $(x_1, \dots, x_n)$.
	This has the required pushforwards because
	\begin{align*}
		\left( \pi_{i,*} \beta \right)_{x_i} & = \bigvee_{x_1, \dots, \cancel{x_i}, \dots, x_n} \beta_{(x_1, \dots, x_n)} \\
		 & = \bigvee_{x_1, \dots, \cancel{x_i}, \dots, x_n} \left( \alpha_{x_1} \wedge \dots \wedge \alpha_{x_n} \right) \\
		 & = \left( \bigvee_{x_1} \alpha_{x_1} \right) \wedge \dots \wedge \alpha_{x_i} \wedge \dots \wedge \left( \bigvee_{x_n} \alpha_{x_n} \right) \\
		 &= \alpha_{x_i},
	\end{align*}
	where the second equation is by distributivity and the third by the normalization assumption $\bigvee_{x_j} \alpha_{x_j} = 1$.
	If $\gamma \in \measfun_\mathcal{B}(X_1 \times \dots \times X_n)$ is another measurement with $\pi_{i,*} \gamma = \alpha_i$ for all $i$, then we have
	\[
		\gamma_{(x_1, \dots, x_n)} \leq \alpha_{x_i}
	\]
	for all tuples $(x_1, \ldots, x_n)$, and therefore
	\[
		\gamma_{(x_1, \dots, x_n)} \leq \alpha_{x_1} \wedge \dots \wedge \alpha_{x_n} = \beta_{(x_1, \dots, x_n)}.
	\]
	Since both $\beta$ and $\gamma$ are partitions of unity in $\mathcal{B}$, these inequalities imply $\gamma = \beta$, which establishes the uniqueness.

	For projectivity, we use the criterion of \cref{projective_equalizer}, which works in the same way as in \cref{projective_projective}.
	Indeed if $\alpha \in \measfun_\mathcal{B}(X)$ and $f, g : X \to Y$ are such that $f_* \alpha = g_* \alpha$, then this means that
	\[
		\bigvee_{x \in f^{-1}(y)} \alpha_x = \bigvee_{x \in g^{-1}(y)} \alpha_x
	\]
	for all $y \in Y$.
	Then since the $\alpha_x$ are pairwise disjoint, we must have $\alpha_x = 0$ if $f(x) \neq g(x)$, and therefore $\alpha$ is indeed supported on $S$.

	For the other direction, we first note that for every Boolean algebra $\mathcal{B}$, the measurements $\alpha \in \measfun_{\mathcal{B}}(X)$, which are partitions of unity in $\mathcal{B}$ indexed by $X$, are in bijection with Boolean algebra homomorphisms $2^X \to \mathcal{B}$.\footnote{Given $\alpha$, the corresponding map $2^X \to \mathcal{B}$ is $A \mapsto \bigvee_{x \in A} \alpha_x$. Note that this bijection is analogous to how a POVM with outcomes in $X$ can either be seen as a tuple of positive operators indexed by $X$ or as a positive operator-valued measure defined on subsets of $X$.}
	Under Stone duality~\cite{Johnstone}, this means that we can identify the measurements with sets of continuous functions,
	\begin{equation}
		\label{stone_duality_measurement}
		\measfun_\mathcal{B}(X) = \mathsf{Stone}(\mathrm{Spec}(\mathcal{B}), X),
	\end{equation}
	where $\mathsf{Stone}$ denotes the category of Stone spaces and continuous maps, $\mathrm{Spec}(\mathcal{B})$ is the Stone space of $\mathcal{B}$ and $X$ is equipped with the discrete topology.

	Second, starting from~\ref{main_thm_classical} amounts to assuming that all finite products and equalizers exist in the category of elements of $\measfun$.
	By \cref{strong_compatibility_product,projective_equalizer} and the uniqueness up to unique isomorphism, we conclude that the forgetful functor from the category of elements back to $\FinSet$ preserves finite products and equalizers.
	Moreover since finite products and equalizers generate all finite limits, we know that the category of elements has all finite limits and that the forgetful functor preserves them, i.e.~that it is a left exact functor.
	It therefore remains to be shown that every left exact functor $\FinSet \to \Set$ is isomorphic to $\mathsf{Stone}(W, -)$ for some Stone space $W$.

	But now it is a special case of a result of Grothendieck~\cite[Corollaire, p.~374]{Grothendieck} that every left exact functor is \emph{pro-representable}, which means representable on the supercategory of pro-objects.
	Since the category of pro-objects on $\FinSet$ is exactly $\mathsf{Stone}$, this means that every left exact functor $\FinSet \to \Set$ is indeed isomorphic to $\mathsf{Stone}(W, -)$ for some Stone space $W$.

	Concerning~\ref{main_thm_onlyclassical}, the implication \ref{main_thm_classical}$\Rightarrow$\ref{main_thm_onlyclassical} is trivial.
	We now finish the proof by showing that~\ref{main_thm_onlyclassical} implies~\ref{main_thm_boolean}.
	To this end, consider the Lawvere theory of Boolean algebras, which we denote by $\LT$.
	By the general theory of Lawvere theories~\cite{HP}, $\LT$ has as objects the free Boolean algebras $2^{2^X}$ on finite sets $X$ and as morphisms the formal opposites of Boolean algebra homomorphisms $2^{2^X} \to 2^{2^Y}$.
	The product-preserving functors $\LT \to \Set$ correspond to Boolean algebras, where to a Boolean algebra $\mathcal{B}$ one associates the functor
	\begin{align}
		\begin{split}
			\label{boolean_lawvere_functor}
			\LT & \longrightarrow \Set, \\
			2^{2^X} & \longmapsto \BoolAlg(2^{2^X}, \mathcal{B}),
		\end{split}
	\end{align}
	where $\BoolAlg$ denotes the category of Boolean algebras and Boolean algebra homomorphisms.
	This is the restriction of the hom-functor $\BoolAlg(-, \mathcal{B})$ to the full subcategory of finitely generated free Boolean algebras.
	By Stone duality, we can equivalently describe $\LT$ as the sets of the form $2^X$ for finite sets $X$ with morphisms given by general maps $2^X \to 2^Y$.
	Under this equivalence, the functor~\eqref{boolean_lawvere_functor} corresponds exactly to~\eqref{stone_duality_measurement} restricted to sets whose cardinality is a power of $2$.
	In other words, if $\FinSet_2 \subseteq \FinSet$ denotes the full subcategory of sets whose cardinality is a power of $2$, then the product-preserving functors $\FinSet_2 \to \Set$ are precisely those functors of the form $\measfun_\mathcal{B}$ restricted to $\FinSet_2$.

	Let now finally $\measfun$ be a strongly classical GMT with $\measfun(\emptyset) = \emptyset$.
	In particular, strong classicality means that $\measfun : \FinSet \to \Set$ is a product-preserving functor.
	Then its restriction to $\FinSet_2$ is also product-preserving, and therefore there is a Boolean algebra $\mathcal{B}$ such that
	\begin{equation}
		\label{restriction_boolean}
		\measfun|_{\FinSet_2} \cong \measfun_\mathcal{B}|_{\FinSet_2}.
	\end{equation}
	We finally argue that this implies $\measfun \cong \measfun_\mathcal{B}$.
	Indeed for every nonempty set $X$, the embedding $X \to 2^X$ given by $x \mapsto \{x\}$ has a retraction.
	Therefore the category of nonempty sets $\FinSet_{\neq \emptyset}$ is the idempotent completion of $\FinSet_2$, and by the general theory of idempotent completions, it follows that the isomorphism~\eqref{restriction_boolean} extends to $\FinSet_{\neq \emptyset}$.
	To conclude, note that the nontriviality assumption implies that $\measfun(\emptyset) = \emptyset$,\footnote{Let $X$ be such that $|\measfun(X)| > 1$. Then the marginalization map $\measfun(\emptyset) = \measfun(\emptyset \times X) \to \measfun(\emptyset) \times \measfun(X)$ can be a bijection only if $\measfun(\emptyset) = \emptyset$.}
	and also that $\mathcal{B}$ is nontrivial (top and bottom element are distinct), which yields the missing $\measfun_\mathcal{B}(\emptyset) = \emptyset$.
\end{proof}

\begin{rem}
	\begin{enumerate}
		\item Instead of proving the implication \ref{main_thm_onlyclassical}$\Rightarrow$\ref{main_thm_boolean} by using the Lawvere theory of Boolean algebras, one can also prove \ref{main_thm_onlyclassical}$\Rightarrow$\ref{main_thm_classical} by results of Trnkov\'a on when product-preserving functors preserve all (finite) limits~\cite[Theorem~1]{Trnkova}.
		\item As the proof shows, the nontriviality restriction is relevant for the implication~\ref{main_thm_onlyclassical}$\Rightarrow$\ref{main_thm_boolean},\footnote{To see that it is actually necessary, let $\measfun : \FinSet \to \Set$ be such that $\measfun(X) = 1$ for all nonempty sets $X$. Then $\measfun(\emptyset)$ can be arbitrary set, and $\measfun$ is product-preserving and thus strongly classical, but it is not projective and hence not isomorphic to $\measfun_\mathcal{B}$ for any Boolean algebra $\mathcal{B}$.} but not for the equivalence \ref{main_thm_classical}$\Leftrightarrow$\ref{main_thm_boolean}.
	\end{enumerate}
\end{rem}

\printbibliography

\end{document}